\pdfoutput=1
\documentclass[runningheads]{llncs}

\usepackage{clrscode}
\usepackage{verbatim}
\usepackage{amsmath}
\usepackage{amsthm}
\usepackage{graphicx}
\usepackage{wrapfig}
\usepackage{subfig}
\usepackage{color}

\renewcommand{\Do}{\textbf{do}\addtocounter{indent}{1}}
\renewcommand{\putindents}{\ifnum\value{thisindent}>0\>\addtocounter{thisindent}{-1}\putindents\fi}
\renewcommand{\Else}{\kill\addtocounter{indent}{-1}\liprint\textbf{else} \> \addtocounter{indent}{1}}
\renewcommand{\Then}{\textbf{then}\addtocounter{indent}{1}}

\newcommand{\inv}[1]{\textbf{(i#1)}}

\newcommand{\father}{\func{f}}
\newcommand{\depth}{\func{d}}

\definecolor{myBlue}{rgb}{0.5,0.5,1}
\definecolor{myRed}{rgb}{0.9,0.3,0.1}
\definecolor{myYellow}{rgb}{0.9,0.9,0}
\definecolor{myGreen}{rgb}{0.1,1,0}

\title{Online validation of the $\pi$ and $\pi'$ failure functions
	}

\begin{document}
\sloppy

\author{Pawe\l{} Gawrychowski
	\and Artur Je\.z \and \L{}ukasz Je\.z
	}
\institute{Institute of Computer Science,
	University of Wroc{\l}aw,\\
	\email{gawry@cs.uni.wroc.pl}, \email{ \{aje,lje\}@ii.uni.wroc.pl}
}

\maketitle


\begin{abstract}
Let $\pi_w$ denote the failure function of the Morris-Pratt algorithm for a word
$w$.
In this paper we study the following problem: given an integer array $A[1 \twodots n]$, is there
a word $w$ over arbitrary alphabet $\Sigma$ such that $A[i]=\pi_w[i]$ for all $i$?
Moreover, what is the minimum cardinality of $\Sigma$ required?
We give a real time linear algorithm
for this problem in the unit-cost RAM model with $\Theta(\log n)$
bits word size.
Our algorithm returns a word $w$ over minimal alphabet such that
$\pi_w = A$ as well and uses just $o(n)$ words of memory.
Then we consider function $\pi'$ instead of $\pi$
and give an online $\mathcal O (n \log n)$ algorithm for this case.
This is the first polynomial algorithm for online version of this problem.
\end{abstract}


\section{Introduction}
The Morris-Pratt algorithm~\cite{MP}, first linear time pattern matching algorithm, is well known for
its simple and beautiful concept. It simulates a forward-prefix-scan DFA for 
pattern matching~\cite{AG-book} by using a carefully chosen \emph{failure
function} $\pi$, also known as a \emph{border array}.
The algorithm utilizes
values of $\pi$ for all prefixes of the pattern. It
behaves like the automaton in the sense that it reads each symbol of the text
once and simulates the automaton's transition.
The amortized time per transition is constant, and the required values
of the prefix function can be calculated beforehand in linear time in a similar fashion.

The failure function itself is of interest as, for instance, it captures all the
information about periodicity of the word. Hence it is often used in word
combinatorics and numerous text algorithms,
see~\cite{AG-book,Crochemore-Rytter-book,jewels}.
The Morris-Pratt algorithm has many variants. In particular,
the Knuth-Morris-Pratt algorithm~\cite{KMP} works in exactly the same manner, but uses a slightly different
failure function, namely the \emph{KMP array} $\pi'$ (or \emph{strong failure function}).
The time bounds for KMP algorithm are precisely
the same as for MP algorithm, but KMP has smaller upper bound on time spent processing a single letter --- for KMP
this bound is $\mathcal O(\log m)$,
whereas for MP it is $\mathcal O(m)$, where $m$ denotes the length of the
pattern.

We investigate the following problem: given an integer array $A[1 \twodots n]$, is there
a word $w$ over an arbitrary alphabet $\Sigma$ such that $A[i]=\pi_w[i]$ for all $i$, where $\pi_w$
denotes the failure function of the Morris-Pratt algorithm for the word $w$. If so, what is the minimum cardinality of
the alphabet $\Sigma$ over which such word exists?
%
%
Pursuing these questions is motivated by the fact that in word combinatorics one
is often interested only in values of $\pi_w$ for every prefix of a word $w$
rather than $w$ itself. Thus it makes sense to ask if there is a word $w$ that
admits $\pi_w=A$ for a given array $A$. Validation of border arrays is also an
important building block of many algorithms that generate all valid border
arrays~\cite{DuvalLecroqLefebvreladnykod,firstlinearverify,countingdistinct}.

We are interested in an \emph{online}
algorithm, i.e. one that receives the input array values one by one, and 
is required to output the answer after reading each such single value.
The maximum time spent on processing a single piece of input is the \emph{delay}
of the algorithm. When the delay is constant, we call the algorithm 
\emph{real time}.
%
%
This and similar problems were addressed before by other researchers. Recently
a linear online algorithm for (closely related) \emph{prefix array} validation
has been given~\cite{Crochemore}.
%
%
A simple linear online algorithm for $\pi$ validation is
known~\cite{DuvalLecroqLefebvreopiprim}, though it has a $\min(n,|\Sigma|)$
delay. Authors of~\cite{DuvalLecroqLefebvreopiprim} were unaware that in this
case $|\Sigma| = \mathcal O (\log n)$~\cite{countingdistinct},
hence the delay of this algorithm is in fact logarithmic.

We provide an online real time algorithm working in unit-cost RAM model (i.e. we assume
that words consisting of $\Theta(\log n)$ bits can be operated on in constant
time) and using $\mathcal O(n\log\log n)$ bits. We show that $\Omega (n)$ bits of space
are necessary if the input is read-once only.

Then we turn our attention to $\pi'$. There is an offline linear bijective
transformation between $\pi$ and $\pi'$. This transformation can be performed
with access to the arrays only, i.e. with no access to the word itself.
Thus it is possible to check offline whether there exists $w$ such that $A = \pi'_w$
in linear time. The task becomes much harder when an online algorithm is required.
Our online algorithm, which is the first
polynomial algorithm for the problem, has running time $\mathcal O (n \log n)$.

This problem was investigated for a slightly different variant of $\pi'$
and an offline validation algorithm for this variant is known
\cite{DuvalLecroqLefebvreopiprim}. The function $g$ considered therein can be
expressed as $g[n] = \pi'[n-1]+1$. The aforementioned bijection between $\pi$
and $\pi'$ cannot be applied to $g$ as it essentially uses the unavailable value
$\pi[n] = \pi'[n]$. While instances on which the algorithm runs in time
$\Theta(n^2)$ are known, no polynomial upper bound on the algorithm's running
time was provided in~\cite{DuvalLecroqLefebvreopiprim}. The algorithm for online
$\pi'$ validation we provide can be applied to $g$ validation with no changes.




\section{Preliminaries}
\label{sec:preliminaries}
For $w \in \Sigma^*$, we denote its length by $n(w)$ or simply $n$. For
$v,w \in \Sigma^*$, by $vw$ we denote the concatenation of $v$ and $w$. We say
that $u$ is a \emph{prefix} of $w$ if there is $v \in \Sigma^*$, such that
$w=uv$. Similarly, we call $v$ a \emph{suffix} of $w$ if there is
$u \in \Sigma^*$ such that $w=uv$. A word $v$ that is both a prefix and
a suffix of $w$ is called a \emph{border} of $w$. By $w[i]$ we denote the $i$-th
letter of $w$ and by $w[i \twodots j]$ we denote the \emph{subword}
$w[i]w[i+1]\ldots w[j]$ of $w$. We call a prefix (respectively: suffix, border)
$v$ of the word $w$ \emph{proper} if $v \neq w$, i.e. it is shorter than $w$
itself.

For a word $w$ its failure function $\pi_w$
is defined as follows: $\pi_w[i]$ is the length of the longest proper border of $w[1 \twodots i]$
for $i=1, 2 \ldots , n$. By $\pi_w^{(k)}$ we denote $\pi_w$ composed $k$
times with itself, namely $\pi_w^{(0)}[i]:=i$ and $\pi_w^{(k+1)}[i]:=\pi_w[\pi_w^{(k)}[i]]$.
This convention applies to other functions as well.
We omit the subscript $w$ in $\pi_w$, whenever it is unambiguous.

\begin{wrapfigure}{r}{0.52\textwidth}
\vspace{-1cm}
\begin{codebox}
\Procname{$\proc{Compute-$\pi$}(w)$}
\zi $\pi_w[1] \gets 0$ , $k \gets 0$ 
\zi \For $i \gets 2$ \To $n$ \label{pi-for-loop} \Do
\zi \While $k > 0$ and $w[k+1] \neq w[i]$ \Do 
\zi $k \gets \pi_w[k]$ 
\End
\zi \If $w[k+1] = w[i]$ \kw{then} $k \gets k+1$
\zi $\pi_w[i] \gets k$ 
\End
\end{codebox}
\vspace{-1.6cm}
\end{wrapfigure}

We state some simple properties of borders and the prefix function. If $u$, $v$ and $w$ are words,
such that $|u| \leq |v| \leq |w|$ and $v$ is a border of $w$, then $u$ is a border of $v$ if and only if
$u$ is a border of $w$. As a consequence, every border of $w[1 \twodots i]$ has length $\pi_w^{(k)}[i]$ for
some integer $k \geq 0$.

\begin{wrapfigure}{r}{0.5\textwidth}
\includegraphics[width=0.5\textwidth]{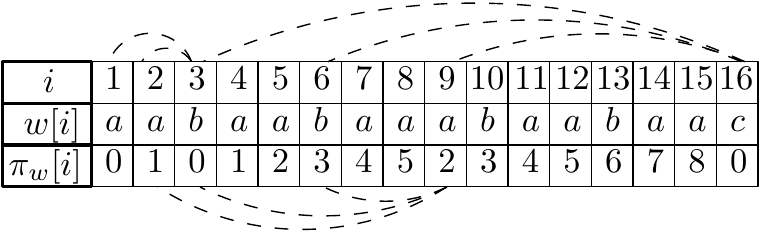}
\caption{Function $\pi$ for word $aabaabaaabaac$. 
The dashed lines represent the
consecutive tries of \proc{Compute-$\pi$} when computing $\pi_w[i]$.
\label{fig:calculating pi example}
}
\vspace{-0.8cm}
\end{wrapfigure}


The \emph{strong failure function} $\pi'$ is defined as follows:
$\pi'_w[n] := \pi_w[n]$, and for $i < n$, $\pi'[i]$ is the length of the longest
(proper) border of $w[1 \twodots i]$, such that $w[\pi'_w[i]+1] \neq w[i+1]$.
If no such border exists, $\pi'[i]=-1$.

It is well-known that $\pi_w$ and $\pi'_w$ can be obtained from one another in linear time, using additional
lookups in $w$ to check conditions of the form $w[i] = w[j]$.
What is perhaps less known, these lookups are
not necessary, i.e. there is a bijection between $\pi_w$ and $\pi'_w$. Values of this function, as well as its
inverse, can be computed in linear time. The correctness of the two procedures
below follows from two (equivalent) observations:
\begin{align} \label{eq:when pi is pi'}
w[i+1] &\neq w[\pi[i]+1] &\iff
\pi[i+1] < \pi[i]+1 &\iff \pi'[i] = \pi[i] \enspace , \\
\label{eq:when pi is not pi'}
w[i+1] &= w[\pi[i]+1] &\iff
\pi[i+1] = \pi[i]+1 &\iff 
\pi'[i] = \pi'[\pi[i]]
	\enspace .
\end{align}

\vspace{-0.5cm}
\begin{minipage}[b]{0.45\linewidth}

\begin{codebox}
\Procname{\proc{Compute-$\pi'$-From-$\pi$}$(\pi)$}
\zi 	$\pi'[0] \gets -1$, $\pi'[n] \gets \pi[n]$ 
\zi 	\For $i \gets 1$ \To $n-1$ \Do 
\zi 		\If $\pi[i+1] = \pi[i]+1$ \Then 
\zi 			$\pi'[i] \gets \pi'[\pi[i]]$
\zi 		\Else $\pi'[i] \gets \pi[i]$
		\End
	\End
\end{codebox}

\end{minipage}
\hspace{0.5cm}
\begin{minipage}[b]{0.45\linewidth}

\begin{codebox}
\Procname{\proc{Compute-$\pi$-From-$\pi'$}$(\pi')$}
\zi $\pi[n] \gets \pi'[n]$ 
\zi \For $i \gets n-1$ \Downto $1$ \Do 
\zi $\pi[i] \gets \max\{\pi'[i],\pi[i+1]-1\}$ 
\End
\end{codebox}
 
\end{minipage}


\section{Online border array validation}
\label{sec:pi verification}
Let $T$ be a graph with vertices $\{ 1, 2, \ldots, n\} $ and directed edges
$\{(k,\pi[k-1]+1) \colon k=2,\ldots,n\}$, see Fig.~\ref{fig:pitree} for an
example. Observe that $T$ is a directed tree: each vertex except the vertex $1$
has exactly one outgoing edge, and since $\pi[i] < i$, the graph is acyclic
and the vertex $1$ is reachable from every other vertex. Thus vertex $1$ is
the root of $T$ and all the edges are directed towards it. Therefore we use
the standard notation $\father[i]$ to denote the unique out-neighbour of $i$
for $i > 1$. We also call $i'$ an \emph{ancestor} of
$i > 1$ if $i' = \father^{(k)}[i]$ for some $k > 0$
(note that $i$ is not its own ancestor).

\begin{wrapfigure}{r}{0.41\textwidth}
\includegraphics[width=0.41\textwidth]{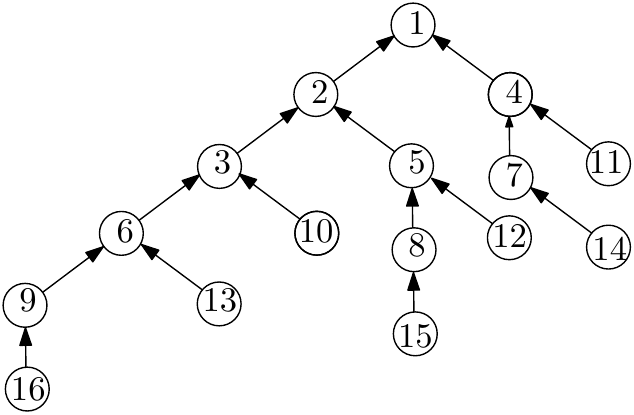}
\caption{An example of tree $T$ for a word from Fig.~\ref{fig:calculating pi example}.
\label{fig:pitree}}
\vspace{-0.87cm}
\end{wrapfigure}

Define an analogue structure $T'$ with $\pi'$ instead of $\pi$:
$\father'[i] = \pi'_w[i-1]+1$. If $\father'[i]=0$ for some $i$
then $i$ has no outgoing edge in $T'$.
Thus $T'$ is a forest and not necessarily a tree.
By~\eqref{eq:when pi is pi'}--\eqref{eq:when pi is not pi'}
$\father'[i]$ can be expressed using $\pi$ and $f$:
\begin{equation}
\label{jak wyglada tatus'}
\father'[i] = \begin{cases}
\father[i] & \text{if } \pi[i] < \father[i] \enspace ,\\
\father'[\father[i]] & \text{if } \pi[i] = \father[i] \enspace .
\end{cases}
\end{equation}

Our approach to online border array validation is as follows: assume there is
a word $w$ admitting $\pi_w=A$ and implicitly construct $T$ for $\pi_w=A$. Then
recover $\pi'_w$ from $\pi_w$ and construct $T'$. Using both $T$ and $T'$,
invalidity of $A$ can be detected as soon as it occurs.


We use the terms \emph{father}, \emph{ancestor}, etc. when referring to $A$,
as it uniquely defines the graphs $T$ and $T'$.
Edges in $T$ reflect the comparisons done by $\proc{Compute-$\pi$}(w)$ for each
position in $w$, or equivalently, each vertex of $T$. In what follows, we
formalise the connection between ancestors in $T$ and (in)equalities between
certain symbols of $w$ that hold under the assumption that $A = \pi_w$.

\begin{lemma}
\label{fact:how_ancestors_look_like}
Suppose that $A = \pi_w$ for some word $w$. Then for $i=2, \ldots, n$ one of the following conditions holds:
\begin{enumerate}
	\item $A[i]=0$ and $w[i] \neq w[i']$ for all ancestors $i'$ of $i$,
	\item $A[i] \neq 0 $ and there exists an ancestor $i' = \pi[i]$ of $i$ such that $w[i] = w[i']$
	and $w[i] \neq w[i'']$ for all interior nodes $i''$ on the directed path from $i$ to $i'$.
\end{enumerate}
\end{lemma}

\begin{proof}

It follows from $\proc{Compute-$\pi$}$: when calculating $\pi_w$,
$\proc{Compute-$\pi$}$ repeatedly checks, whether $w[\pi^{(k)}[i-1] + 1 ] = w[i]$
for successive values of $k$, starting with $k=1$. Precisely, it checks, whether
$w[i'] = w[i]$ for successive ancestors $i'$ of $i$. This test immediately ends when smallest $k$
is found, such that it satisfies $w[\pi^{(k)}[i-1] + 1 ] = w[i]$ or $\pi^{(k)}[i-1]=0$.
In the latter case $w[i] \neq w[i']$ for all ancestors $i'$ of $i$ and $A[i]=0$,
whereas in the former case $A[i] \neq 0 $ and there exists an ancestor $i' = \pi[i]$ of $i$,
such that $w[i] = w[i']$ and $w[i] \neq w[i'']$ for all $i'' \neq i,i'$ on the path from $i$ to $i'$.
\qed
\end{proof}

This follows from the way \proc{Compute-$\pi$} works. Refer to Fig.~\ref{fig:pitree} for an example.
Using $T'$ a slightly stronger statement can be formulated:

\begin{lemma}
\label{pi is usually pi'}
Suppose that $A = \pi_w$ for some word $w$. Then 
either $A[i]= \father[i]$ or $A[i]=0$ or $A[i] = \father'^{(k)}[i]$ for some $k \geq 1$.
\end{lemma}

\begin{proof}
When calculating the candidates for $\pi[i]$ we look at the sequence of values
$\father[i],\father^{(2)}[i],\ldots$ but whenever $w[\father^{(k)}[i]]=w[\father^{(k+1)}[i]]$, we
can safely skip $\father^{(k+1)}[i]$. The largest
$i'$ such that $i' = \father^{(k')}[i-1]+1$ and $w[i'] \neq w[i] $ is $i'=\pi'[i-1]+1=\father'[i]$.
\qed
\end{proof}

The following criteria follow from Fact~\ref{fact:how_ancestors_look_like}: if $A = \pi_w$ for some $w$ then
\begin{eqnarray}
A[1]=0 \label{A is 0}
\label{A_is_an_ancestor}
\quad \text{ and } \forall i > 0 \quad
A[i]> 0 &\Longrightarrow& \exists k \quad A[i] = \father^{(k)}[i] \; , \\
\label{different_letters}
\forall i > 0 \quad
\forall {k>0} \quad A[i] < \father^{(k)}[i] &\Longrightarrow& A[i] \neq A[\father^{(k)}[i]] \enspace .
\end{eqnarray}


Conditions~\eqref{A is 0}--\eqref{different_letters} 
are necessary and sufficient for $A$ to be a valid $\pi$ array \cite{DuvalLecroqLefebvreopiprim}. 
They yield an algorithm for testing whether $A$ is a valid $\pi$ table and calculating
the minimal size of required alphabet \cite{DuvalLecroqLefebvreladnykod}.

Instead of checking whether there is $j$ on the path from $i$ to the root
such that $j = A[i]$ and $j$ is a valid candidate for $\pi[i]$, one can store
all the valid candidates for $i$ at the node $i$ and do the required checks
locally.
It turns out that the sets of valid candidates satisfy
a simple recursive formula \cite{DuvalLecroqLefebvredlugie}:
\begin{equation}
\label{eq:candidates}
\id{cand}[i] = \begin{cases}
\left\{0,\father[i]\right\}\cup(\id{cand}[\father[i]]\setminus \left\{A[\father[i]]\right\}) &\text{ if }A[i]\neq 0\\
\left\{0\right\} &\text{ if }A[i]=0
\end{cases}
\end{equation}

\begin{wrapfigure}{r}{0.6\textwidth}
\vspace{-1.2cm}
\begin{codebox}
\Procname{$\proc{Validate}(A)$}
\zi \If $A[1] = 0$ \Then
\zi \Error $A$ not valid at $1$
	\End
\zi $\id{cand}[1] \gets \emptyset$, $w[1]\gets 1$
\zi \For $p \gets 2$ \To $n$ \Do
\zi		\If $A[p]=0$ \Then
\zi				$\id{alph}[p] \gets \id{alph}[\father[i]]+1 $
\zi				$\id{max-alph} \gets \max(\id{max-alph}, \id{alph}[p])$
\zi				$w[p] \gets \id{alph}[p]$
\zi				$\id{cand}[p] \gets \emptyset$
			\Else	$\id{cand}[p] \gets \id{cand}[\father(p)] \setminus \{ A[\father(p)] \} \cup \father(p)$
\zi				\If $A[p] \notin \id{cand}[p]$ \Then
\zi 				\Error $A$ not valid at $p$
				\End
\zi				$w[p] \gets w[A[p]]$, $\id{alph}[p] \gets \id{alph}[\father[i]]$
		\End	
	\End
\end{codebox}
\vspace{-1cm}
\end{wrapfigure}

Moreover, in~\eqref{eq:candidates}, $\id{cand}[\father'[i]]\setminus
\left\{A[\father'[i]]\right\}$ instead of $\id{cand}[\father[i]]\setminus
\left\{A[\father[i]]\right\}$ can be used, leading to a more sophisticated
algorithm $\proc{Validate}(A)$~\cite{DuvalLecroqLefebvredlugie}. It runs in
linear time and space, and has $\mathcal{O}(\min (n,|\Sigma|))$
delay~\cite{DuvalLecroqLefebvredlugie}. Note that the running time bound is not
obvious, as a set of candidates is kept at each node. It can be bounded by
noting that each valid candidate corresponds to one non-trivial transition of
the automaton recognizing the language $\Sigma^\ast w$, and the number of such
transition is linear, regardless of the size of $\Sigma$~\cite{Simon}.
Moreover, the minimal size of $\Sigma$ is
$\mathcal O (\log n)$~\cite[Th. 3.3a]{countingdistinct},
of which authors of~\cite{DuvalLecroqLefebvredlugie} were unaware.


Let $\depth'$ denote the depth of $T'$. In our algorithm
\proc{Validate-$\pi$-RAM}, we exploit the fact that
$\depth' \in \mathcal O (\log n)$. It follows easily from the following lemma:
\begin{lemma}
\label{pi' cube}
If $A=\pi_w$, then $\father'^{(3)}[i] < \frac{\father'[i]}{2}$ for all $i$.
\end{lemma}

\begin{proof}
First observe that if $\pi^{(2)}_w[n] \geq \frac n 2 -1$ then $w[\pi_w[n]+1]=w[\pi_w^{(2)}[n]+1]$. Indeed, both
$n - \pi[n]$ and $n-\pi^2[n]$ are periods of $w[1 \twodots n]$. Since their sum is at most $\frac n 2 + 1 + \frac n 2 = n +1$,
by periodicity lemma also $\gcd(n - \pi[n], n - \pi^{(2)}[n])$
is a period, hence $\pi[n] - \pi^{(2)}[n]$ is a period as well. Therefore
$$
w[\pi_w^{(2)}[n]+1 ] = w[\pi_w^{(2)}[n]+1 + \pi[n] - \pi^{(2)}[n]] = w[\pi_w[n]+1] \enspace .
$$

Now consider a path in $T$ such that $i_{k} = \father[i_{k+1}]$ and $\father'^{(3)}[i_\ell]=i_1$.
We aim at showing that $\father'^{(3)}[i_\ell] < \frac{\father'[i_\ell]}{2}$.

We begin with observation that $w[i_2] \neq w[i_1]$, i.e. that $\pi[i_2] \neq i_1 $, by~\eqref{eq:when pi is pi'}.
Assume the opposite. In particular $\father'[i_2] < i_1$. Let $i = \father'^{(2)}[i_\ell]$.
Then $\father'[i] = i_1$. By definition this means that $\pi'[i-1] = i_1 - 1$, thus $w[1\twodots i_1-1]$ is a border
of $w[1 \twodots i -1]$ and $w[i]\neq w[i_1]$. But $w[1\twodots i_2-1]$ is a border
of $w[1 \twodots i -1]$ as well and $w[i_2] = w[i_1] \neq w[i]$. Since $i > i_2$ it holds that $\pi'[i] \geq i_2$,
contradiction. We conclude that $w[i_2] \neq w[i_1]$.

Take $n = i_3 - 1$ and suppose that $\pi^2[n] \geq \frac {n}{2}-1$.
Then by first paragraph $w[\pi^2[n]+1] = w[\pi[n]+1]$.
But $\pi^2[n]+1 = \father^2[i_3]=i_1$ and $\pi[n]+1 = \father[i_3]=i_2$.
Hence $w[i_2] = w[i_1]$, contradiction. Therefore $\pi^2[n] < \frac {n}{2}-1$.
That is $\pi^{(2)}[i_3-1] < \frac{i_3-1}{2}-1$, hence
$i_1 = \pi^{(2)}[i_3-1] + 1 < \frac{i_3-1}{2}$.

Since both $\father$ and $\father'$ are monotone functions then
$\father^{(2)}[i_3] =i_1$ implies $\father'^{(2)}[i_3] \leq i_1$.
Therefore $\father'[i_\ell] \geq i_3$, since $\father'^{(2)}[i_\ell] = i_1$.
We conclude that 
$$
\father'^{(3)}[i_\ell] = i_1 < \frac{i_3-1}{2} < \frac{\father'[i_\ell]}{2}
$$
which ends the proof. \qed
\end{proof}

\begin{wrapfigure}{r}{0.6\textwidth}
\vspace{-1.25cm}
\begin{codebox}
\Procname{$\proc{Validate-$\pi$-RAM}(A)$}
\zi \If $A[1] = 0$ \kw{then} \Error $A$ not valid at $1$
\zi $\id{Bcand}[1] \gets \emptyset$, $w[1]\gets 1$, $\depth[1] \gets \depth'[1] \gets 1$
\zi \For $i \gets 2$ \To $n$ \Do
\zi		$\depth[i] \gets \depth[\father[i]] +1$, $\depth'[i] \gets \depth'[\father[i]]$
\zi		\If $A[i] = \father[i]$ \kw{then}	$\depth'[i] \gets \depth'[i]+1$
\zi		\If $A[i]=0$ \Then
\zi				$w[i] \gets \id{alph}[i] \gets \id{alph}[\father[i]]+1 $
\zi				$\id{max-alph} \gets \max(\id{max-alph}, \id{alph}[i])$
\zi				$\id{Bcand}[i] \gets \emptyset$
\zi			\Else $j \gets \func{LA}(i,\depth[i] - \depth[A[i]]-1)$
\zi				\If $A[i] \neq \father[j]$ or $\depth'[j] = \depth'[A[i]]$ \Then 
\zi					\Error $A$ not valid at $i$
				\End
\zi				\If $\id{Bcand}[i][\depth'[A[i]]] = 0$ \Then 
\zi					\Error $A$ not valid at $i$
				\End
\zi				$\id{Bcand}[i] \gets \id{Bcand}[\father[i]]$
\zi				$\id{Bcand}[i][\depth'[\father[i]]] \gets 1$
\zi				$\id{Bcand}[i][\depth'[A[\father[i]]]] \gets 0$
\zi				$w[i] \gets w[A[i]]$, $\id{alph}[i] \gets \id{alph}[\father[i]]$
		\End	
	\End
\end{codebox}
\vspace{-1.25cm}
\end{wrapfigure}
By Lemma~\ref{pi is usually pi'}, either $A[i]=\father[i]$ or $A[i]$ is
$i$'s ancestor in $T'$. Using these observations, $\depth'$ values of
all valid candidates for $\pi[i]$, except $\father[i]$, can be stored at vertex
$i$ instead of the candidates themselves.
Vertex $i$ stores them encoded in a bit vector $\id{Bcand}[i]$. The $j$-th bit
of $\id{Bcand}[i]$ is set to $1$ if
there is a valid candidate $i'$ for $\pi[i]$ with $\depth'[i']=j$, and $0$
otherwise. The depths in $T'$ can be encoded in a bit vector using only 
a constant number of machine words. To validate $A[i]$, we check if
$\depth'[A[i]]$-th bit of $\id{Bcand}[i]$ is set to $1$ and perform two further
tests: first, we check if $A[i]$ is an ancestor of $i$ in $T$, and then if
there is a valid candidate $j$ for $\pi[i]$ among the ancestors of $i$ in $T$
such that $\depth'[i']=\depth'[A[i]]$. Clearly, $A[i]$ is valid if and only if
all three tests are successful.

To perform these tests efficiently, we use a data
structure~\cite{Alstrup00improvedalgorithms}, working in RAM model, that
supports the \emph{level ancestor query} in any tree:
$\func{LA}(i,\Delta)$ returns the ancestor $j$ of $i$ that is $\Delta$ levels
above $i$. The data structure also supports addition of new leaves and takes
only constant time per any of these two operations. To use this data structure,
each node needs two additional fields, $\depth$ and $\depth'$.

\begin{theorem}\label{RAM correctness}
\proc{Validate-$\pi$-RAM} works in linear time with constant delay. It uses
linear number of machine words.
\end{theorem}

\begin{proof}
First of all note that the calculations of $\depth'$ are proper,
as by \eqref{jak wyglada tatus'} and easy induction it holds that
\begin{equation}
\label{jak wyglada d'}
\depth'[i] =
	\begin{cases}
	1 &\text{ if } i=1 \\
	\depth'[\father[i]] &\text{ if } i>1 \text{ and } A[i] = \father[i]\\
	1+ \depth'[\father[i]] &\text{ if } i>1 \text{ and } A[i] \neq \father[i]
	\end{cases}
\end{equation}

Update of the information kept in $\id{Bcand}$ is correct, as it follows
directly the recurrence relation for sets of valid candidates~\eqref{eq:candidates}.

The memory usage is obvious, as only a constant number of machine words per
node is used. The same applies to the running time: only a constant number of
operations per letter of input is performed, and the level ancestor query takes
only constant time. Additional cost of maintaining the data structure for those
queries is only a constant per position read.

Now we inspect the problem of checking whether $A[i]$ is the unique node
among ancestors $j$ of $i$ such that $\depth'[j]=\depth'[A[i]]$
that is a valid candidate for $\pi[i]$.
We show that the valid candidate is the one among those vertices which has the largest depth in $T$.
Consider $j$ that is a valid candidate for $\pi[i]$ and $j'$ such that $\father[j']=j$.
Suppose that $\depth'[j]=\depth'[j']$. Then by \eqref{jak wyglada d'}
$A[j']=j$. But then by \eqref{A_is_an_ancestor} $A[i] \neq j'$
implies $A[i] \neq A[j'] = j$, contradiction. So $\depth[j'] \neq= \depth'[j]$.
Note that by \eqref{jak wyglada d'} the ancestors of $i$ of the same depth in $T'$
are consecutive nodes on the path from $i$ to the root.
We conclude that the valid candidates $j$ for $\pi[i]$ of fixed $\depth'[j]$ is the one
among the ancestors of $i$ of fixed $\depth'$ that has the largest depth in $T$.

Consider the path from $i$ to the root in $T$. It consists of blocks of nodes
such that $\pi[j]=\father[j]$. By~\eqref{jak wyglada d'}, positions in each
block have the same $\depth'$, and $\depth'$ decreases by $1$ when on the
block's end. Suppose that $j=\father[j']$ is a valid candidate and that it is
not the first vertex in its block. By~\eqref{different_letters}, if
$A[i] \neq j'$, then $A[i] \neq A[j']=j$, contradiction. So if $j$ is a valid
candidate, it must be the first vertex in its block.


To prove the correctness, it is enough to show that \proc{Validate-$\pi$-RAM}
correctly recognises whether $A[i]$ is one of $i$'s valid candidates for
$\pi[i]$. First of all, if $A[i]\neq 0$, it is checked whether $A[i]$ is an
ancestor of $i$ in $T$: using level ancestor query the ancestor at the depth of
$A[i] $ is recovered, if it is different than $A[i]$, $A$ is invalid at $i$.
To check whether $A[i]$ is the first vertex in the block of vertices of the
same $\depth'$, recover the previous vertex on the path from $i$ to the root 
and check if it has different $\depth'$. If not, $A$ is rejected. If $A[i]$ is
the first in its block, check whether there is a valid candidate in that block.
If so, then clearly $A[i]$ is the one.
\qed
\end{proof}

Both $\proc{Validate-$\pi$}(A)$ and algorithm
from~\cite{DuvalLecroqLefebvredlugie} use a linear number of machine words,
i.e. $\Theta(n\log n)$ bits. It can be shown that at least $\Omega(n)$ bits
are necessary in the \emph{streaming setting}, i.e. when successive input values
are given one by one and cannot be re-read.


\begin{theorem}\label{streaming bound}
Deterministic streaming verification of $\pi$ or $\pi'$ array of length $n$ requires $\Omega(n)$ bits of memory.
\end{theorem}

\begin{proof}
Since $0$ and $\pi[i-1]+1$ are always valid values for $\pi[i]$, then there are at least $2^{n-1}$
$\pi$ arrays of length $n$. Assume that $\frac{n}{2}-2$ bits
are enough. It means that there are two different valid prefixes of $\pi$ arrays of length $n/2$ ($\pi_1$ and $\pi_2$)
after reading which the algorithm is in the same state.
Let $i\leq\frac{n}{2}$ be any index at which those two prefixes
differ: $\pi_1[i]<\pi_2[i]<i$. We append $0,1,2,3,\ldots,i,\pi_2[i]+1,0,0,\ldots,0$ to both of them
($\pi_2[i]+1$ is at position $\frac n 2 + i +2$).
After reading both input sequences the algorithm is in the same state for both of them. But
exactly one of the resulting arrays is valid:
consider $\pi_1[\frac n 2 + i +2]$ --- just before reading $\pi_2[i]+1<i+1$, all possible candidates for
$\pi_1[\frac n 2 + i +2]$ are $i+1,\pi[i]+1,\pi^2[i]+1$, ... .
And $i+1 > \pi_2[i]+1 > \pi_1[i]+1$, thus the sequence is invalid.
What is left to show is that $\pi_2[i]+1$ is a valid candidate for $\pi_2$ at position $\frac n 2 + i +2$.
Since $(\pi_2[i]+1) \neq i+1$ then it should hold that $\pi_2[i]+1 \neq \pi_2[i+1]$, by \eqref{A_is_an_ancestor}.
But $\pi_2[i+1] = 0$ and $\pi_2[i]+1 > 0$.

So, we get a contradiction: the algorithm is in the same state in both cases and thus cannot be correct.
\qed
\end{proof}

Although we do not know if $O(n)$ are enough, we are able to show that the total memory usage
can be reduced to just $O(n\log\log n)$ bits, i.e. a sublinear amount of machine words.
The algorithm remains real time, the delay is still constant and so the total running time
is linear.

In order to reduce the memory usage, we cannot store the values of $\father[i]$ for each $i$, as this may use $\Omega(n \log n )$ bits
(consider text $a^n$). Hence  we store $\father'[i]$. It turns out that there cannot be
too many different large values of $\father'$ and that they can be all stored (using some clever encoding)
in $\mathcal O (n\log\log n)$ bits.

To implement this approach we adapt $\proc{Validate-$\pi$}$ so that it uses $\father'$ instead of $\father$.
Let us take a closer look at $\proc{Validate-$\pi$}(A)$.
Assuming that is has already processed the prefix
$A[1\twodots i-1]$, we know it is a valid $\pi$ array.
Thus we may calculate its corresponding $\pi'$ array, denoted by
$A'[1\twodots i-1]$. By Lemma~\ref{pi is usually pi'} the values
of $A[1],A[2],\ldots,A[i-2]$ are not needed: whenever $A[i]\neq 0$ and $A[i] \neq A[i-1]+1$,
we check whether $A[i]$ is among the ancestors of $i$ in $T'$.
If so, we verify whether no lower ancestor of $i$ in $T'$ was assigned the same letter.

In order to present the details of this algorithm we need to understand the underlying combinatorics of $\pi'$.
The following lemma shows that different large values of $\pi'$ cannot be packed too densely, which allows
us to store information about different positions in a more concise way.

\begin{lemma}\label{few different}
Let $k\geq 0$ and consider a segment of $2^k$ consecutive indices in the $\pi'$ array. At most
$48$ different values from $[2^k,2^{k+1})$ occur in such a segment.
\end{lemma}

\begin{proof}
First notice that each $i$ such that $\pi'[i]>0$ corresponds to a non-extensible occurrence of prefix
$w[1 \twodots \pi'[i]]$, i.e. $\pi'[i]$ is maximal among $j$ such that
$w[1 \twodots j]$ is a suffix of $w$ but $w[j+1]\neq w[i+1]$.

If $k < 2$ then the claim is trivial. So let $k'=k-2\geq 0$
and assume that there are more than $48$ different values from $[4 \cdot 2^{k'},8 \cdot 2^{k'})$
occurring in a segment of length $2^k$.
Then more than $12$ different values from $[4 \cdot 2^{k'},8 \cdot 2^{k'})$
occur in a segment of length $2^{k'}$.
Split this range into three subranges $[4 \cdot 2^{k'},5 \cdot 2^{k'})$,
$[5 \cdot 2^{k'},6 \cdot 2^{k'})$ and $[6 \cdot 2^{k'},8 \cdot 2^{k'})$.
Hence at least $5$ different values from one of those such subrange $[\ell,r)$ occur in this segment,
for some $\ell,r$. Note that $r\leq\frac{3}{2}\ell-2^{k'}$.
Let these $5$ different values occur at positions $p_1<\ldots <p_5$.
Consider the sequence $p_i-\pi'[p_i]+1$ for $i=1,\ldots,5$: these are the beginnings
of the corresponding non-extensible prefixes.
In particular all these elements are pairwise different.
Each sequence of length $5$ contains a monotone subsequence of length $3$.
We consider the cases of increasing and decreasing sequence separately:
\begin{enumerate}
\item there exists $p_{i_1}<p_{i_2}<p_{i_3}$ in this segment such that
$p_{i_1} - \pi'[p_{i_1}] > p_{i_2} - \pi'[p_{i_2}] > p_{i_3} - \pi'[p_{i_3}]$.
\begin{figure}
\centering
\includegraphics{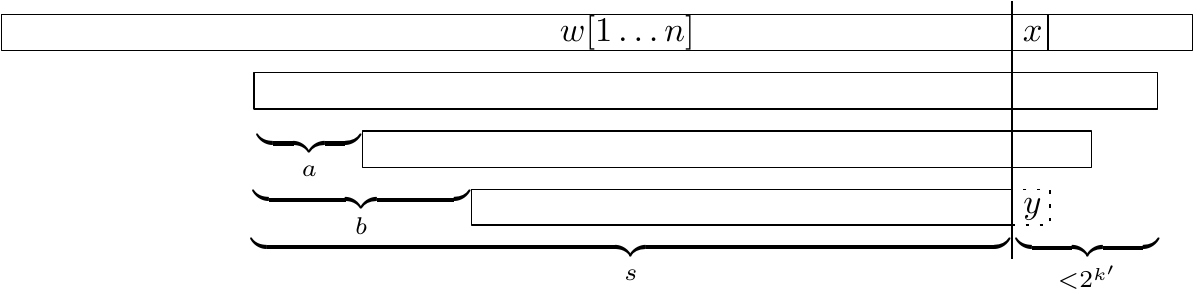}
\caption{Proof of Lemma~\ref{few different}, increasing sequence.}\label{decreasing}
\end{figure}

Both $a$ and $b$ are periods of $w[1\twodots s]$ (see Fig.~\ref{decreasing}).
As $s \geq \ell$ and $a < b\leq r-\ell$,
condition $r\leq\frac{3}{2}\ell-2^{k'}$
ensures that $a,b\leq\frac{s}{2}$. Thus by periodicity lemma $b-a$ is also a period of $w[1\twodots s]$.
But then $w[p_1+1]=w[p_1+1+b-a]$, so $x=y$ making the value of $\pi'[p_1]$ incorrect.

\item there exists $p_1<p_2<p_3$ in this segment such that
$p_{i_1} - \pi'[p_{i_1}] < p_{i_2} - \pi'[p_{i_2}] < p_{i_3} - \pi'[p_{i_3}]$.

\begin{figure}
\centering
\includegraphics{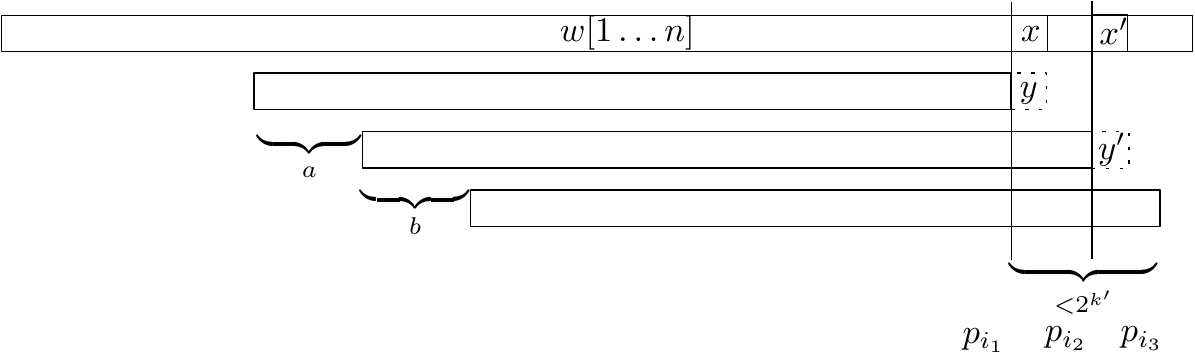}
\caption{Proof of Lemma~\ref{few different}, decreasing sequence.}\label{increasing}
\end{figure}

Let $s_1=\pi'[p_{i_1}]$ and $s_2=\pi'[p_{i_2}]$ (see Fig.~\ref{increasing}). Then $s_i \geq \ell$ by assumption.
Moreover $a + b\leq r-\ell+2^{k'}$, so condition $r\leq\frac{3}{2}\ell-2^k$
ensures that $a + b \leq \ell /2 \leq s_i /2 $ for $i=1,2$.
As $s_1\neq s_2$, there are two subcases:

\begin{enumerate}

\item $s_1<s_2$: then $b$ is a period of $w[1\twodots s_1+1]$ and $w[s_1+1]=w[s_1+1-b]$.
Because $a$ is a period of $w[1\twodots s_1]$ and $a,b\leq \frac{s_1}{2}$, $w[s_1+1-b]=w[s_1+1-b-a]$.
As $b$ is a period of $w[1\twodots s_1+1]$, $w[s_1+1-b-a]=w[s_1+1-a]$.
Thus $x=y$, making the value of $\pi'[p_1]$ incorrect.

\item $s_1>s_2$: similarly $a$ is a period of $w[1\twodots s_2+1]$ and $w[s_2+1]=w[s_2+1-a]$.
Because $b$ is a period of $w[1\twodots s_2]$ and $a,b\leq \frac{s_2}{2}$, it holds that $w[s_2+1-a]=w[s_2+1-a-b]$.
As $a$ is a period of $w[1\twodots s_2+1]$, $w[s_2+1-a-b]=w[s_1+1-b]$.
So $x'=y'$, making the value of $\pi'[p_2]$ incorrect.
\qed
\end{enumerate}
\end{enumerate}
\end{proof}

This observation on the combinatorial property of $\pi'$ allows
us to state the promised algorithm with constant delay and $\mathcal O (n\log\log n)$ bits of memory usage.

\subsubsection*{Organisation of memory}
For each position $x$ we would like to store: $\father'[x]$, $w[x]$,
the list of all its ancestors and a bit vector of flags denoting those of ancestors that were
valid candidates for $\pi[x]$ see Fig.~\ref{vertex_struct}.
This requires $\Theta (b^2(x))$ bits, where $b(x) = \log \father'[x]$, which is too much for us.
Observe that all this information, except $w[x]$, depends solely on $\father'[x]$:
it contains the list of ancestors in $T'$ and the list of valid candidates for $\pi[x]$,
which depends only on the list of candidates for $\pi[\father'[x]]$, by~\eqref{eq:candidates}.
Thus for each position $x$ we store $w[x]$ and $b(x)$ instead of $\father'[x]$.
For technical reasons we also store an $k$ such that $A[x] = \father'^{(k)}[x]$
and a flag denoting whether $A[x]=0$.
Since the alphabet size is $\mathcal O (\log n)$~\cite{countingdistinct}
and the number of ancestors in $T'$ is logarithmic (by Lemma~\ref{pi' cube}),
the memory usage is $O(n\log\log n)$ bits.

\begin{figure}
\vspace{-0.5cm}
\centering
\includegraphics{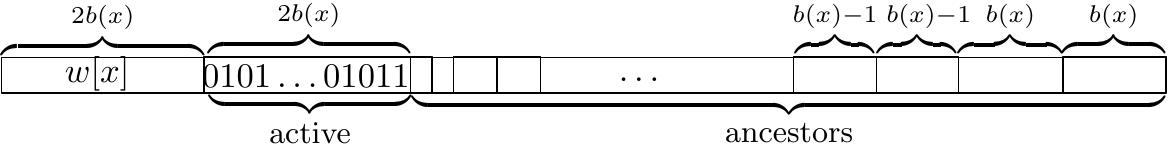}
\caption{Information needed for a single value of $\father'[x]$.}\label{vertex_struct}
\vspace{-0.5cm}
\end{figure}

We now estimate the size of the space needed for a single value of $\father'[x]$.
By Lemma~\ref{pi' cube}, $\father'^{(3)}[x] < \frac{\father'[x]}{2}$.
Hence the binary encodings of at most $2$ ancestors of $x$ in $T'$ have length exactly $k$.
We reserve exactly two chunks of $k$ bits for each possible candidate of length $k = 1 , \ldots , b(x)$. 
Hence the total number of bits associated with a single position $\father'[x]$ is $O(b^2(x))$.

The bit vector is the $\id{Bcand}$ bit vector of \proc{Validate-$\pi$-RAM} with sole exception:
we keep $\id{Bcand'}[x] = \id{Bcand}[x] \setminus \{0, A[x] \}$ instead of $\id{Bcand}[x]$.
It is easy to see that it satisfies the relation:
\begin{equation}
\label{eq:candidates2}
\id{Bcand'}[x] = \begin{cases}
\left\{ x \right\} \cup (\id{Bcand'}[\father'[x]] \setminus \left\{ A[x] \right\}) &\text{ if }A[x]\neq 0\\
\emptyset  &\text{ if }A[i]=0
\end{cases}
\end{equation}
We use the table $\id{Bcand'}[\father'[x]]$
in order to check whether $A[x]$ is a valid candidate for $A[x]$. Clearly $0$ is a valid candidate
and $A[\father'[x]]$ is not.

For each possible value of $\alpha = b(x)$,
we allocate $48\frac{n}{2^\alpha)}$ blocks of $O(\alpha^2)$ bits.
These are called $\alpha$-blocks.
Each group of $48$ consecutive $\alpha$-blocks corresponds to a segment
$[\ell 2^\alpha, (\ell +1 ) 2^\alpha)$ of the input.
A single $\alpha$-block consists of the value of $\father'$ and the information
described in the previous paragraph.
The amount of allocated memory is
$\mathcal O(\sum_\alpha \frac{n}{2^\alpha}\alpha^2)=\mathcal O(n)$ bits.


\subsubsection*{Validating input}
Consider $x$ such that $b(x)=\alpha $ and $x \in [\ell 2^\alpha, (\ell +1 ) 2^\alpha)$.
We want to check if $A[x]$ is a valid candidate for $\pi[x]$: values $A[x-1]+1$ and $0$
are always valid; otherwise $A[x]$ has to be an ancestor of $x$ in $T'$, by Lemma~\ref{pi is usually pi'}.

To retrieve the information associated with position $x$, we first calculate $\alpha = b(x)$ in $\mathcal O(1)$ time.
This gives the offset in memory where $\alpha$-blocks for a segment $[\ell 2^\alpha, (\ell +1 ) 2^\alpha)$ are stored.
Then we look up $\alpha$-block number of $x$, which allows accessing information on ancestors of $x$.
We check if $A[x]$ is among them. There are two positions in the block on which
$A[x]$ may be stored; their distances from the beginning of the record can be calculated using
a constant number of arithmetic operations. If $A[x]$ is one of the ancestors,
we check if it is a valid candidate for $\pi[x]$ by a look-up in the bit vector.
Clearly it takes only $\mathcal O(1)$ time.

Then we store the values $\father[x+1]$ and $\father'[x+1]$, which are needed for $x+1$.
The latter value can be calculated easily using~\eqref{jak wyglada tatus'}.

\subsubsection*{Update}
Suppose we add a new position $x$.
Then we search the corresponding range of $48$ $\alpha$-blocks to see if $\father'[x]$ is already stored.
If $\father'[x]$ is not present, we reserve the next unoccupied block for
$\father'[x]$. The list of ancestors can be recreated from list of ancestors of $\father'^{(2)}[x]$.
The list of valid candidates is the list of valid candidates for $\father'^[(2)][x] $
with the flag for $A[\father'[x]]$ set to $0$, according to~\eqref{eq:candidates2}.
This can be done in constant time,
as for $\father'[x]$ we store the number $k$ such that $\father'^{(k)}[\father'[x]]=A[x]$.
Then we set the flag denoting whether $A[x]=0$, and calculate $k$ such that $A[x] = \father'^{(k)}$.
Both operations can be done in constant time.
\subsubsection*{Running time; lazy copying}
We cannot copy eagerly,
as there might be as much as $\log n$ machine words to be copied.
We use a lazy approach instead, keeping a list of memory regions
(each possibly consisting of many machine words)
that should be copied. After processing each index,
we copy a constant number of words from this list.

Assume that $\alpha b(x)$ words need to be copied for a single index $x$,
and after processing each position $\beta$ words from the list are copied.
Whenever there are many elements on the list, we choose the one corresponding to the smallest value of $b(x)$.
For that we keep a separate list for each possible value of $b(x)$
and a bit vector indicating which lists are non-empty.
To make things working, we need to ensure that information on $x$ is eventually copied, but we cannot
start copying it too early, as it might block copying other information which we will need much sooner.
Therefore we start copying the words associated with $x$ such that $\alpha = b(x)$
and $x \in [\ell \cdot 2^ \alpha, (\ell+1) \cdot 2^ \alpha )$
after processing $(\ell+1) \cdot 2^ \alpha$.
We later show that we finish before copying all those words before processing $(\ell+1) \cdot 2^\alpha$.

Ensuring that information on $x$ is copied, but not too early, is a little tricky.
For each possible value of $b(x)$ we keep another list of waiting memory regions, a bit vector indicating
which of those lists are non-empty, and a bit vector indicating which of those lists should be merged with
the corresponding lists of regions ready to be copied.

After processing an index divisible by $2^{b(x)}$, we would like to move all waiting memory regions from the
corresponding lists to the lists of regions that should be copied. Although moving one list take only
a constant time if we use a simple single-linked implementation, there can be many lists to take care of.
Hence we just mark those non-empty list that should be moved. Then whenever we need to extract an element
with the smallest value of $b(x)$, we look at both lists of elements ready to be
copied and those waiting lists.
If this list is marked as one to be copied, we move all its elements before adding a new one.

We know that the total amount of information to be copied is linear (in terms of code words),
so we have enough time do it. Our concern is that we must
bound the delay between processing an index and copying all its associated information:

\begin{lemma}\label{fast_copy}
All information associated with $x$ such that $\gamma = b(x) $
and $x \in [ \ell 2^ \gamma , (\ell + 1) 2^ \gamma )$
is successfully copied before processing $(\ell+2) 2^ \gamma$.
\end{lemma}

\begin{proof}

We know that the total number of machine words to be copied is linear,
so there is enough time do it. The delay is our only concern.

We call the memory chunks for $x$'s such that $b(x)=\gamma$ the \emph{$\gamma$-chunks}.
Suppose that there are several procedures responsible for copying memory chunks:
procedure \proc{Copy-$\gamma$} is responsible for copying $\gamma$-chunks.
After processing each position we copy $\alpha \beta$ memory words,
where $\beta$ is an appropriate constant which we will calculate later,
and there are $\alpha \gamma$ machine words to be copied for $x$.
Imagine we are given $\alpha \beta$ \emph{credit} after each position;
note though that this is a worst-case analysis, not an amortised one.

We run the procedure \proc{Copy-$\gamma$}, where $\gamma$ is the smallest value such that
some $\gamma$-chunks are to be copied. If \proc{Copy-$\gamma$} did not use all the credit,
we repeat the process (for larger $\gamma$) until we run out of it.

Consider the procedure \proc{Copy-$\gamma$} and the interval $[\ell 2^{\gamma},(\ell+1) 2^{\gamma})$.
All the information to be copied while this interval is read is ready before we process
position $\ell 2^{\gamma}$. Since \proc{Copy-$\gamma$} can use all the credit for this interval
that was not used by \proc{Copy-$(\gamma-1)$}, \ldots , \proc{Copy-$0$},
we subtract an upper bound on the credit used by them from the $\alpha \beta 2^\gamma$ credit
we are given for processing this interval.

Let $c_\gamma$ be the maximal credit used by \proc{Copy-$0$}, \ldots , \proc{Copy-$(\gamma)$}
on an interval of length $2^\gamma$ assuming they do not run out of credit.
Then $\alpha \beta 2^\gamma - 2c_{\gamma -1}$ is the credit available to \proc{Copy-$(\gamma)$} on this interval:
the credit used by \proc{Copy-$0$}, \ldots , \proc{Copy-$(\gamma-1)$} on interval of length $2^\gamma$
consists of credit used by them on two intervals of length $2^{\gamma-1}$. On the other hand the credit released is
$\alpha \beta 2^\gamma$. 

We give a recursive formula for $c_\gamma$. We can assume that $c_0=1$.
Then by Lemma~\ref{few different} \proc{Copy-$\gamma$} copies at most
$48$ records, each of them consisting of at most $\alpha \gamma$ machine words.
Let us upper bound the credit used by \proc{Copy-$0$}, \ldots , \proc{Copy-$(\gamma-1)$} on an interval of length $2^\gamma$.
Divide this interval into two sub-intervals of length $2^{\gamma-1}$.
Then by definition on each of them \proc{Copy-$0$}, \ldots , \proc{Copy-$\gamma$} used at most $c_{\gamma-1}$ of credit and so

$$
c_\gamma = 2 c_{\gamma-1} + 48 \alpha \gamma:
$$

By standard techniques, $c_\gamma = \mathcal O (2^\gamma)$. Let $\beta$ be such that $c_\gamma \leq \beta 2^\gamma$.
We now show by induction on $\gamma$ that 
$48 \alpha \gamma \leq \alpha \beta 2^\gamma - 2c_{\gamma -1}$,
i.e. the credit available to \proc{Copy-$(\gamma)$} is enough to pay for the copying it should perform.
It trivially holds for $\gamma = 0$. 
Now consider $\gamma +1$:
$$
\alpha \beta 2^{\gamma+1} - 2c_{\gamma} =
\alpha \beta 2^{\gamma+1} - c_{\beta+1} + 48 \alpha \gamma \geq 48 \alpha \gamma \enspace .
$$
So there is enough time to copy all the required information.
\qed
\end{proof}

As a consequence, each list contains at most $48$ elements, so the total size of the memory required
to perform lazy copying is just $\mathcal O(\log^2 n)$.

To make use of this lazy copying, we must remember about a few details.
Just before we put pointer to the block of memory corresponding to $x$
on the list of chunks that should be copied, we copy the value $\father'[x]$
and light a special flag meaning that its contents is still processed.
When we want to extract some information about $x$ from the memory,
it might happen that the information is not copied yet.
In such case we look at the block of memory corresponding to its $\father'$.
If it is also not ready yet, we look at its $\father'^{(2)}$, and so on.
Lemma~\ref{pi' cube} guarantees that $\father'^{(5)}[x]<\frac{\father'[x]}{4}<\frac{x}{4}$.
By Lemma~\ref{fast_copy}, the information associated with $\father'^{(5)}[x]$
is copied after processing at most $\frac{x}{2}$ indices, well before considering $x$.
Thus a constant number of lookups is enough to get to an ancestor
who stores a complete list of his own ancestors.
This gives an algorithm both time optimal and using a sublinear amount of machine words:

\begin{theorem}
$\pi$ array of length $n$ can be validated online in real time using
$O (n \frac{\log \log n}{\log n}) = o(n)$ machine words of $\Theta(\log n)$ bits.
\end{theorem}


\section{Online strict border array validation}
\label{sec:pi' verification}
\begin{wrapfigure}{r}{0.4\textwidth}
\vspace{-0.8cm}
	\centering
	\includegraphics{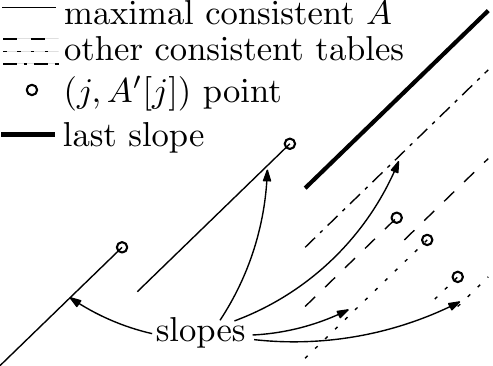}
	\caption{Maximal consistent function.\label{maximal}}
\vspace{-2.5cm}
\end{wrapfigure}
While we already know a simple algorithm \proc{Compute-$\pi$-From-$\pi'$},
this algorithm is not online.
Therefore we need another approach.
In general, we assume that $A'$ is a valid $\pi'_w$ for some word $w$, recover $\pi_w$ out of $A'$
and then run \proc{Validate-$\pi$} on $\pi_w$ to calculate the word $w$ and the minimal size of the required
alphabet. In the following we present algorithm \proc{Validate-$\pi'$} that performs this task.

\subsubsection*{Overview of the algorithm}
Imagine the array $A'$ as the set of points $(i,A'[i])$.
We say that a table $A[1 \twodots n+1]$ is consistent with
$A'[1 \twodots n]$ iff the following two conditions hold
\begin{enumerate}
	\item[\inv1] $A[1\twodots n+1] = \pi_w[1 \twodots n+1]$ for some word $w[1\twodots n+1]$;
	\item[\inv2] $\pi'_w[1\twodots n]= A'[1\twodots n]$;
\end{enumerate}
The algorithm keeps a \emph{maximal function consistent with} $A[1 \twodots n+1]$ for $A'[1 \twodots n]$,
i.e. the one satisfying the condition:
\begin{enumerate}
	\item[\inv3] for every $B[1 \twodots n+1]$ such that $B$ is consistent with $A'[1\twodots n]$
	it holds that $A[j] \geq B[j]$ for $j= 1 , \ldots , n+1$ \enspace ,
\end{enumerate}
see Fig.~\ref{maximal}.
Note that after reading $n$ input symbols we recover $A[1\twodots n+1]$.
To express shortly that $A$ is a maximal candidate we use a notation $A[1 \twodots m] \geq B[1 \twodots m]$
to denote that $A[j] \geq A[j]$ for $j = 1 , \ldots ,m$.
The invariant of the algorithm is that the computed function $A$ satisfies conditions \inv1--\inv3.

We think of $A$ as a collection of maximal \emph{slopes}:
a set of indices $i,i+1, \ldots, i+j$ is a slope if $A[i+k] = A[i]+k$ for $k = 1, \ldots, j$.
Note that, by~\eqref{eq:when pi is pi'}, $A[i+j+1] \neq A[i+j]+1$ implies that 
$A[i+j]=A'[i+j]$.
When a new letter is read, we have to update $A$
or claim that $A'$ is invalid. It turns out that only the last slope has to be updated.

It can be shown that \inv1--\inv3 imply a stronger property, which is essential
for our analysis.
\begin{lemma}
\label{inv4}
Let $A[1\twodots n+1]$ be a maximal function consistent with $A'[1\twodots n]$ and $B[1\twodots n+1]$
be consistent with $A'[1 \twodots n]$. Let $i$ be the first position
of the last slope of $A$. Then $A[1 \twodots i-1] = B[1 \twodots i-1]$.
\end{lemma}

\begin{proof}
If there is only one slope, there is nothing to prove. If there are more, consider $i-1$
--- the last element on the second to the last slope.
Since this is the end of a slope, then by~\eqref{eq:when pi is pi'} $A'[i-1] = A[i-1]$.
On the other hand, consider $B[1 \twodots n+1]$ as in the statement of the lemma.
By \inv3 it holds that $A[i-1] \geq B[i-1]$. Thus
$$
A'[i-1] \leq B[i-1] \leq A[i-1] = A'[i-1] \enspace ,
$$
hence $B[i-1] = A[i-1]$. Let $B[1 \twodots n+1] = \pi_{w'}[1 \twodots n+1]$.
Using \proc{Compute-$\pi$-From-$\pi'$} we can uniquely recover $\pi_{w'}[1\twodots i-1]$
from $\pi_{w'}'[1\twodots i-1]$ and $\pi_{w'}[i-1]$.
But those values are the same for $A[1 \twodots i-1]$, thus 
$A[1\twodots i-1] = \pi_{w'}[1\twodots i-1]$.
\qed
\end{proof}

\begin{wrapfigure}{r}{0.4\textwidth}
	\vspace{-0.8cm}
	\centering
	\includegraphics{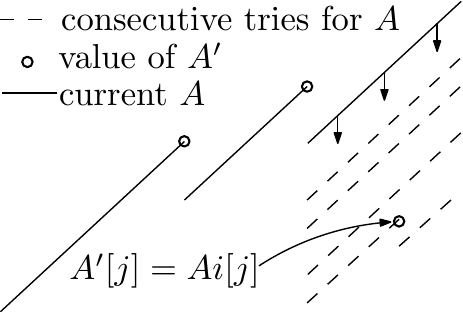}
	\caption{Candidate decreasing.\label{fig:adjustinglastslope}}
\vspace{-0.9cm}
\end{wrapfigure}
When new value $A'[n]$ is read, it may happen that $A[1\twodots n]$ does not satisfy
\inv2: by \eqref{eq:when pi is pi'}--\eqref{eq:when pi is not pi'}
$A'[n]=A[n]$ or $A'[j]=A'[A[n]]$ holds if $A[1 \twodots n+1]$ is consistent with $A'[1 \twodots n]$.
Then we adjust the values of $A$ on the last slope.
Suppose that there is some other valid candidate function $A_1[1 \twodots n+1]$.
Since by invariant \inv3 for $j = 1 , \ldots , n$ it holds that $A[1 \twodots n] \geq A_1[1 \twodots n]$
and for $j \in [1 \twodots n]$ it holds that $A[j] > A_1[j]$.
In order to replace $A$ by some valid candidate
we have to decrease it at some point, see Fig.~\ref{fig:adjustinglastslope}.
Let $i$ denote the beginning of the last slope of $A$.
By Lemma~\ref{inv4} $A[1 \twodots i-1] = A_1[1 \twodots i-1]$.
So in each step we decrease the value of $A[i]$ by the smallest offset,
so that $A[1 \twodots n] \geq A_1[1 \twodots n]$ is kept.
It can happen that at some index $j$ it holds that $A'[j]>A[j]$. Then $A_1 < A'[j]$
and as $A_1$ is chosen arbitrarily, 
$A'$ is invalid.
On the other hand it may turn out that $A'[j]=A[j]$.
In such case $A_1[j]=A[j]$ 
and we shorten the last slope:
by~\eqref{eq:when pi is pi'} $A'[j]=A[j]$ implies $A[j+1] < A[j]+1$.

\subsubsection*{Information stored}
The algorithm stores:
\begin{itemize}
	\item the input read so far, i.e. a prefix of $A'$,
	\item suffix tree for $A'$, created online \cite{McCreight,Ukkonen},
	\item number $n$ denoting the number of read values of $A'$,
	\item first position $i$ on the last slope,
	\item table $A[1 \twodots i-1]$, those values are fixed,
	\item candidate value $A[i]$, it may be changed later.
\end{itemize}
The algorithm also uses implicit values of $A[i+j]=A[i]+j$ for $j=1, \ldots , n-i+1$.
These values do not need to be stored in the memory.


\proc{Validate} is run for $A\pi[1 \twodots i-1]$ (or $A[1 \twodots i]$, if $A[i]=0$),
i.e. on values of $A$ not changed later.
Since by \inv1 $A$ is a valid border array \proc{Validate} cannot call an error. 
It is run in order to calculate the minimal size of the alphabet,
letters of the word $w$ and a set of valid candidate for $\pi[i]$.
Note that even valid candidates for $n+1$ can be calculated as by \eqref{eq:candidates}
set of candidates for $\pi[i]$ is expressed in terms of candidates of $\father[i]<i$.
Moreover, since $A$ is fixed for $\father[i]$, the set of candidates for $i$ is calculated just once.
Since \proc{Validate} is an online algorithm,
we may feed it with values of $A$ as soon as they are stored.

\subsubsection*{Update and adjusting the last slope}
When next value $A'[n]$ is read we check whether $A'[n]=A'[A[n]]$.
Otherwise $A$ ceased to be a valid border table or $i$ last slope
is defined improperly. Hence we adjust $A$ on the current last slope.

When adjusting the last slope we aim at satisfying two conditions
\begin{equation}
A'[j] < A[j] \label{A jest mniejsze niz A'} \text{ and } 
A'[j] = A'[A[j]]  \enspace , 
\end{equation}
for each $j \in [i \twodots n]$. These conditions are checked by two queries:
the \emph{height-query} returns the smallest $j \in [i \twodots n]$ such that
$A'[j] \geq A[j]$; the \emph{value-query} answers whether
$A'[i \twodots n] = A'[A[i] \twodots A[i] + (n-i)]$.
Note that the second query is just a short way of asking whether for $j \in [i \twodots n]$
the second condition of \eqref{A jest mniejsze niz A'} holds, as $A[j] = A[i] + (j-i)$.

We ask both queries until the height query returns no index and value query returns {\bfseries yes}.If the height-query returns an index $j$ such that $A'[j]> A[j]$, then we
reject the input and call an error. If $A'[j] = A[j]$ then we check (naively) whether
$A'[i\twodots j-1] = A'[A[i] \twodots A[i] + (j-i-1)]$.
If not, we reject. If it holds then a new slope $[i \twodots j]$ and a new last slope $[j+1 \twodots n]$
are created, we store values $A[i\twodots j]$,
$i$ is set to $j+1$ and we set $A[i]$ to the largest possible candidate value for $\pi[i]$.
Then we continue adjusting.

If value-query answers { \bfseries no} then we set the value of $A[i]$ to the next largest
valid candidate value for $\pi[i]$ and continue with adjusting.
If $A[i] = 0$ then there is no such candidate value and we reject. 
%
\begin{lemma}
\label{niezmienniki pi'}
After reading a valid strong prefix array $A'[1 \twodots n]$ \proc{Validate-$\pi'$}
satisfies conditions \inv1--\inv3. Otherwise \proc{Validate-$\pi'$} rises an error.
\end{lemma}

\begin{proof}
We proceed by induction on $n$.
If $n=0$, then clearly $A[1]=0$ and all the invariants trivially hold and $A'$ is a valid $\pi'$ array.

Whenever a new symbol is read, then \proc{Validate-$\pi'$} checks 
the second condition of \eqref{A jest mniejsze niz A'} for $j=n$. If it holds, then no changes are needed because:

Condition \inv1 holds trivially: we implicitly set $A[n+1]=A[n]+1$,
which is always a valid value for $\pi[n+1]$.

Condition \inv2 holds: since $A'[n]=A'[A[n]]<A[n]$, by \eqref{eq:when pi is not pi'} $A[n]$ is properly defined.

Condition \inv3 holds: consider any $B[1 \twodots n+1]$ consistent with $A'[1 \twodots n]$.
By induction assumption \inv3 holds for $A[1 \twodots n]$ hence $B[n] \leq A[n]$.
Therefore $B[n+1] \leq B[n]+1 \leq A[n]+1=A[n+1]$. Thus $A$ is still maximal.

Suppose that the second condition of \eqref{A jest mniejsze niz A'} is not satisfied for $j=n$.
Then the algorithm starts adjusting $A$.
Consider first a case when there is a function consistent with $A'[1 \twodots n]$.
We show that during the adjusting \inv3 holds --- i.e.
for every $A_1[1\twodots n+1]$ consistent with $A'[1\twodots n]$ it holds that
$A_1[1\twodots n+1] \leq A[1\twodots n+1]$, even though $A[1 \twodots n+1]$
may not be consistent with $A'[1 \twodots n]$.
Moreover, during the adjustments \inv1 is preserved. In the end we show that when
the adjustments stop then \inv2 is satisfied.

So let $A_1$ be as described earlier.
By Lemma~\ref{inv4}, $A[1 \twodots i-1] = A_1[1 \twodots i-1]$.
The algorithm repeatedly checks conditions \eqref{A jest mniejsze niz A'}
using height-query and value-query. Suppose that height query returns $j$ such that $A[j]< A'[j]$.
Then, since \inv3 is satisfied, $A_1[j] \leq A[j] <A'[j]$, i.e. $A_1$ is not a valid $\pi$ table. Contradiction.
So suppose that height-query returned $j$ such that $A[j] = A'[j]$. Then
similarly, by \inv3,
$A_1[j] \leq A[j] = A'[j]$ but as $A_1$ is a valid $\pi$ table, it also holds that $A_1[j]\geq A'[j]$.
So $A[j]=A'[j]=A_1[j]$, and by \eqref{eq:when pi is pi'} $j$ is an end of a slope for $A_1$.
Since $A_1$ is chosen arbitrarily, this holds for any table $B$ consistent with $A'$.
Then, by Lemma~\ref{inv4}, $A[i \twodots j]=A_1[i \twodots j]$.
Since for $j \in [i \twodots j-1]$ it holds that $A[j] > A'[j]$ then, by\eqref{eq:when pi is not pi'},
$A[j]$ and $A'[j]$ must satisfy equation $A'[j] = A'[A[j]]$.
If this equation is not satisfied by some $j$ then clearly we reject the output, as $A_1$ is invalid as well.
Then the algorithm sets $i$ to $j+1$ and sets $A[i]$ to the largest possible candidate value for $\pi[i]$
smaller than $A[i-1]+1$. Note that $A_1[i] < A[i-1]+1$ by \eqref{eq:when pi is pi'}.
The implicit values $A[j]$ for $j \in [i+1 \twodots n]$ satisfy $A[j] = A[i] + (j-i)$ .
Thus \inv1 is still satisfied --- all the changed values of $A$
are valid $\pi$ candidate values at the respective positions.
Since $A_1$ is a valid $\pi$ table,
$$
A_1[j] \leq A[i] + (j-i) \leq A[i] + (j-i) = A[j] \enspace ,
$$
i.e., $A$ still satisfies \inv3.

If for all $j \in [i \twodots n]$ it holds that $A[j] > A'[j]$, then the algorithm asks the value query.
Suppose it is not satisfied for some $j$, i.e. $A'[j] \neq A'[A[j]]$.
Hence $A_1[j] \neq A[j]$, as by \eqref{eq:when pi is pi'}--\eqref{eq:when pi is not pi'}
either $A_1[j]=A'[j]<A[j] $ or $A'[j] \neq A'[A_1[j]]$ . We show that $A_1[i] \neq A[i]$ ---
assume otherwise and consider the smallest $j'$ such that $A_1[j'+1] \neq A[j'+1]$.
Since \inv3 holds
$$
A_1[j'+1] < A[j'+1]=A[j']+1=A_1[j']+1 \enspace .
$$
Hence $A_1[j'+1] < A[j']+1$. By \eqref{eq:when pi is pi'}, $A'[j']=A_1[j']=A[j']$, contradiction.
On the other hand, \proc{Validate-$\pi'$} sets $A[i]$ to next largest valid candidate for $\pi[i]$.
So $A[i] \geq A_1[i]$ still holds. Since the implicit values satisfy $A[j] = A[i] + (j-i)$
for $i \in [i+1 \twodots n ]$, then also $A_1[j] \leq A_1[i]+(j-i) \leq A{i} +(j-i) = A[j]$.
So \inv3 still holds for $A$. Note again that $A[i \twodots n+1]$ were all
assigned valid candidates for $\pi$
at their respective positions, so \inv1 still holds.

Now we show that when both conditions
\eqref{A jest mniejsze niz A'}
are satisfied (i.e. when all the adjustments are finished)
invariant \inv2 is satisfied as well.
Note that~\eqref{eq:when pi is pi'}
and~\eqref{eq:when pi is not pi'} give the following formulas for $\pi'$ in terms of $\pi$:
$$
	\pi'[j]= \begin{cases}
		\pi[j] &\text{ if } j \text{ is the last element of some slope} \enspace ,\\
		\pi'[\pi[j]] &\text{ in other cases} \enspace .
	\end{cases}
$$
Those formulas are verified for $j$ when $A[j]$ is stored.
For $j$ such that $A[j]$ are implicit values, i.e. such that $j$ is on the last slope of $A$,
those are verified by value-query.
Hence when all adjustments are finished, \inv2 holds.

Suppose now that the input is not a valid strong prefix table.
We show that if \proc{Validate-$\pi'$} accepts the input then $A'$ is valid.
Since \inv1 was preserved during the adjustments, $A$ is a valid $\pi$ table.
Moreover, for each position $j$ conditions \eqref{A jest mniejsze niz A'}
are satisfied --- the adjustments of the slopes ends when they are satisfied for each position.
So $A'$ is a valid candidate for $\pi'_w$ such that $\pi_w = A$.
\qed
\end{proof}

\begin{theorem}\label{theorem validate pi'}
\proc{Validate-$\pi'$} correctly computes $\pi_w$, such that $A' = \pi'_w$
and calculates the required minimum size of the alphabet.
\end{theorem}

\begin{proof}
By Lemma~\ref{niezmienniki pi'} \proc{Validate-$\pi'$} raises an error if the input
table is invalid and otherwise returns the size of the alphabet.
So we need to show that the size of the alphabet is computed correctly.

By invariant \inv1, $A[1 \twodots n+1] = \pi_w[1 \twodots n+1]$ for some word $w$
and by invariant \inv2 $\pi'_w = A'$.
Word $w[1 \twodots i-1]$ is in fact created by \proc{Validate-$\pi$}.
So two questions remain:
is the alphabet required for $w$ minimal for $A'$ and is the answer given by
\proc{Validate-$\pi$} really the alphabet needed for $w[1 \twodots n]$,
(as \proc{Validate-$\pi$} is run on the prefix $A[1\ \twodots i-1]$ only).

Suppose first that $A[i]>0$.
Thus \proc{Validate-$\pi$} run on
$A[1 \twodots n]$ returns the same size of required alphabet as run on
$A[1 \twodots i-1]$, as new letters are needed only when $A[j]=0$ at some position,
and $A[j]>0$ for $j$ on the last slope. So consider any $B[1 \twodots n+1]$
consistent with $A'[1 \twodots n]$. Then by Lemma~\ref{inv4} $B[1 \twodots i-1] = A[1 \twodots i-1]$.
Clearly \proc{Validate-$\pi$}$B[1 \twodots i-1]$ reports the required size of the alphabet
not larger than \proc{Validate-$\pi$}$B[1 \twodots n]$. 
So indeed $A$ does not require larger alphabet than $B$.


Suppose now that $A[i]=0$. Then for any $B[1 \twodots n+1]$ consistent with $A'[1 \twodots n]$, by \inv3
it holds that $0 \leq B[i] \leq A[i]=0$.
Since $A[j]>0$ for $j>i$, the same argument as previously works.
\qed
\end{proof}

\subsubsection*{Answering queries}
Consider first height-query.
The idea is that if $A'[j'] - A'[j] > j' - j > 0$ then $j$ cannot be an answer to height-query.
Using this observation a list of possible answers can be kept and quickly updated.
\begin{lemma}
\label{lem:pi is pi' constant query}
Answering all height-queries can be done in amortised linear time.
\end{lemma}

\begin{proof}

\begin{figure}
\centering
\includegraphics{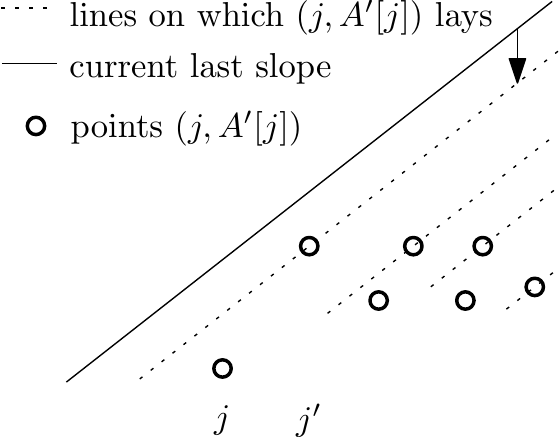}
\caption{Answering height-queries.\label{endofslopes}}
\end{figure}

Consult Fig.~\ref{endofslopes}.
The idea is as follows: consider any two indices $j$, $j'$ such that
$$
A'[j'] - A'[j] > j' - j > 0 \enspace .
$$
We denote this relation by $j \prec j' $ and say that $j'$ \emph{dominates} $j$.
Then $j$ cannot be an end of any slope, if $A[j]=A'[j]$ then
$$
A[j'] \leq A[j] + (j'-j) = A'[j] + (j'-j) < A'[j] + A'[j'] - A'[j] = A'[j'] \enspace ,
$$
contradiction.
Note that $j \prec j'$ and $j' \prec j''$ implies $j \prec j''$: clearly $j < j' < j''$ and
$$
A'[j'] - A'[j] > j' - j \text{ and } A'[j''] - A'[j'] > j'' - j'
$$
summed up implies that
$$
A'[j''] - A'[j] > j'' - j \enspace .
$$

This can be reformulated in terms of height-queries: if $j < j'$ and $A'[j'] - A'[j] > j' - j$
then $A'[j] \geq A[j] $ implies $A'[j'] > A[j']$, i.e. that the instance is invalid.
Hence we need not keep track of $j$ as a potential answer to the height-query.
It is enough to keep a list of positions $j_1 < j_2 < \ldots < j_k$
such that $j_i \not \prec j_\ell$ for all $i,\ell$
and $j_\ell$ dominates all $j \in [j_{\ell -1}+1 \twodots j_\ell-1]$.

When the query is asked we check if $A[j_1] \leq A'[j_1]$.
We show that evaluating this expression for other values of $j$ is not needed.
Suppose that $A'[j] \geq A[j]$ for some $j \in [j_{\ell -1}+1 \twodots j_\ell-1]$.
Then since $j_\ell$ dominates $j$ it holds that $A'[j_\ell] > A[j_\ell]$.
Suppose now that $A'[j_\ell] \geq A[j_\ell]$ for $j_\ell > j_1$.
Then since $j_1 < j_\ell $ and $j_\ell$ does not dominate $j_1$ it holds that
$A'[j_\ell] - A'[j_1] \leq j_\ell - j_1$.
As $j_1$ and $j_\ell$ are on the last slope then $A[j_\ell] = A[j_1] + (j_\ell - j_1)$, hence
$$
A[j_1] = A[j_\ell] - (j_\ell - j_1) \leq A[j_\ell] -( A'[j_\ell] - A'[j_1]) \leq A'[j_1] \enspace ,
$$
so $j_1$ is a proper answer to the height-query.
So the height-query is answered in constant time.

We demonstrate that all updates of the list $j_1 , \ldots , j_k$ can be done
in $\mathcal O (n)$ time.
When new position $n$ is read, we update the list
by successively removing $j_\ell$'s dominated by $n$
from the end of the queue.
By routine calculations, if $n$ dominates $j_\ell$, then it dominates $j_{\ell+1}$:
\begin{align*}
A[j_{\ell+1}] - A[j_{\ell}] \leq j_{\ell+1} - j_\ell \enspace , \\
A[n] - A[j_{\ell}] > n - j_\ell
\end{align*}
imply
$$
A[n] - A[j_{\ell+1}] > n - j_{\ell+1} \enspace .
$$
So we have to remove some suffix of the kept list of $j$'s.

Suppose that $j_\ell , \ldots , j_k$ were removed. Then $j_\ell , \ldots , j_k \prec n$,
so $j \prec n$ for each $j \in [j_{\ell-1} + 1 \twodots j_k-1]$.
Moreover $j_{\ell-1} \not \prec n$ and thus also $j \not \prec n$ for $j = j_1, \ldots j_{\ell-1}$.

As each position enters and leaves the list at most once,
the time of update is linear.
\qed
\end{proof}

To answer value-queries efficiently we construct online a suffix tree \cite{McCreight,Ukkonen}
for the input table $A'[1 \twodots n]$. Answering the value queries can be done by dividing the query into
two sub-queries, one is checked naively the other by traversing up in the suffix tree.
It is possible to amortise both sub-queries.
\begin{lemma}
\label{lem:value query}
Answering all value-queries can be done in $\mathcal O (n \log n)$ time.
\end{lemma}

\begin{proof}
First of all, a suffix tree is constructed online \cite{McCreight,Ukkonen} for the input table $A'[1 \twodots n]$.
This takes $\mathcal O (n \log n)$ time --- the construction is linear in length of the word but logarithmic
in the size of alphabet. Since $A'$ may have values up to $n-1$ we have to include the logarithmic factor.

Fix an index $i$ and consider all value-queries asked
while $i$ was considered $i$ as the beginning of the last slope.
The set of valid candidates for $\pi[i]$ is of size $\mathcal O (\log n)$:
by Lemma~\ref{pi is usually pi'} only one candidate is not of the form $\father'^{(j)}[i]$
for some $j$ and by Lemma~\ref{pi' cube} there are only $\mathcal O (\log n)$
positions of such form.
Hence for a fixed position $i$ there are $\mathcal O (\log n)$ value-queries asked.
Suppose that the queries were asked for candidates $A[i]-\ell_1$, \ldots, $A[i]-\ell_{k_i}$.
We show that the query about $A[i] - \ell_j$ can be answered in $\mathcal O (\ell_j)$.
Then
$$
\mathcal O (\sum_{m=1}^{k_i} \ell_m) = \mathcal O (\sum_{m=1}^{k_i} \ell_{k_i}) =
\mathcal O (k_i \ell_{k_i}) = \mathcal O (\ell_{k_i} \log n) \enspace .
$$
Then we sum over all possible $i$ and show that $\sum_i \ell_{k_i} \leq n$
hence the result can be upper bounded $\mathcal O (n \log n)$.

\begin{figure}
	\includegraphics{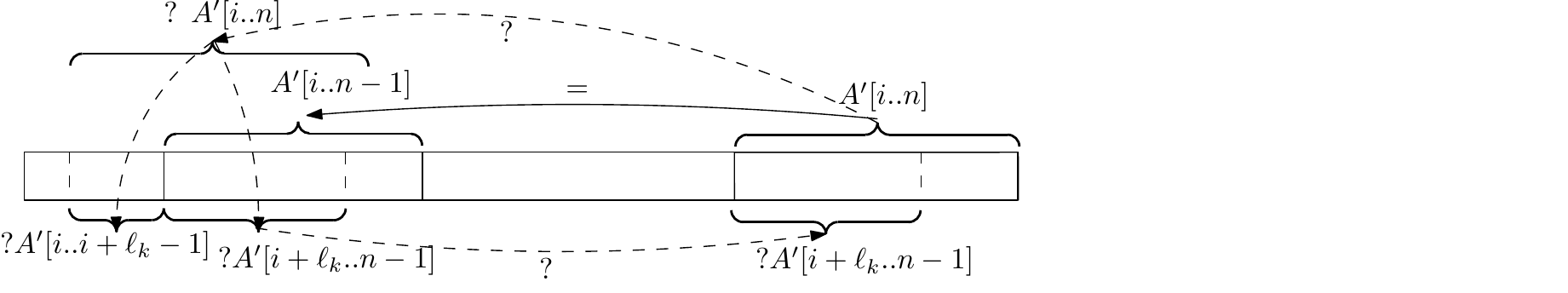}
	\caption{The subqueries of value-query. Solid lines represent the known
	equalities between fragments of $A'$ table, dashed lines represent the tests.
	The test and equalities should be read according to the arrow-heads.
	\label{valuequery}}
\end{figure}

Suppose that we check whether $A'[i \twodots n] = A'[A[i]-\ell_j \twodots A[i]-\ell_j + (n-i)]$.
If $\ell_j \geq (n-i)$, this can be done naively in $\mathcal O (n-i) = \mathcal O (\ell_j)$ time.
If $\ell_j < (n-i)$, we divide the tests into three subtests, see Fig.~\ref{valuequery}.
On one hand, we check naively whether
$$
A'[i \twodots i + \ell_j-1] = A'[A[i] - \ell_j \twodots A[i]-1] \enspace ,
$$
this is done in $\mathcal O (\ell_j)$ time.
We also naively check $A[n] = A[i] + (n - i - \ell_j)$ in constant time.
Finally we check whether
$$
A'[i + \ell_j \twodots n-1] = A'[A[i] \twodots A[i] + (n - i - \ell_j-1)] \enspace .
$$
We show how to perform this test efficiently.
Since $A[1 \twodots n]$is a valid candidate for $A'[1 \twodots n-1]$
and $i$ was the first position on the last slope of $A$, then
$$
A'[i \twodots n-1] = A'[A[i] \twodots A[i]+(n-i)-1] \enspace ,
$$
by \eqref{eq:when pi is not pi'}. Therefore
$$
A'[A[i] \twodots A[i] + (n - i - \ell_j)-1] = A'[i \twodots n - \ell_j-1] \enspace .
$$
So it is enough to check whether
$$
A'[i + \ell_j \twodots n-1] = A'[i \twodots n - \ell_j-1] \enspace ,
$$
i.e. whether $A'[i + \ell_j \twodots n-1]$ is a prefix of $A'[i \twodots n-1]$.
This is easily done using suffix trees:
wlg. we may assume that the check is made before $A[n]$ was added to the suffix tree.
We enrich the suffix trees: each node has a pointer to its father. This does not increase the built-time.
Then we go to the vertex corresponding to suffix $A'[i \twodots n-1]$ and traverse the tree $\ell_j$
letters up. We return whether suffix $A'[i + \ell_j \twodots n-1]$ ends in this node.
Traversing up costs $\mathcal O (\ell_j)$ time.

Consider again all the queries asked at position $i$ for valid candidates for $\pi[i]$ equal
$A[i]-\ell_1 > \ldots > A[i]-\ell_{k_i}$. Then $A[i]$ was replaced by $A[i] - \ell_{k_i}$,
hence also the value of $A[n]$ was replaced by $A[n] - \ell_{k_i}$.
Since $A[n]$ increases by at most $1$ when $n$ increases by $1$
thus $\sum_{i=1}^n \ell_{k_i}$ is linear in $n$. Therefore:
\qed
\end{proof}

\subsubsection*{Running time}
Construction of the suffix tree and answering value queries takes $\mathcal O (n \log n)$ time,
answering height-queries takes $\mathcal O (n)$ time. Running $\proc{Validate}(A')$ takes $\mathcal O (n)$ time.
Therefore the algorithm runs in $\mathcal O (n \log n)$ time.

\subsection*{Pseudocode of \proc{Validate-$\pi'$}}

\begin{codebox}
\Procname{\proc{Validate-$\pi'$}$(A')$}
\zi $A[1] \gets 0$, $\id{i} \gets 1$, $n \gets 0$
\zi \Repeat
\zi		$n \gets n+1$
\zi		\If $A'[n] > A[n]$ \kw{then} \Error $A'$ is not valid at $n$
\zi		$A[n+1] \gets A[n]+1$
\zi		\Repeat
\zi			$\id{change} \gets \const{false}$
\zi			\If there is $j \in [i \twodots n]$ such that $A'[j] > A[j]$ \Then 
\zi				\Error $A'$ is not valid at $n$
			\End
\zi			\If 	there is $j' \in [\id{f} \twodots n]$ such that $A'[j'] = A[j']$ \Then 
\zi				let $j \gets$ minimal such $j'$
\zi				\If $A'[i \twodots j-1] \neq A'[A[i] \twodots A[i] + (j-i-1)]$ \Then
\zi					\Error $A'$ is not valid at $n$
				\End
\zi 				run \proc{Validate-$\pi$}$(A)$ on positions $i\twodots j$
\zi				\If $A[j+1]=0$ \Then
\zi					run \proc{Validate-$\pi$}$(A)$ on position $j+1$
				\End
\zi				\For $m \gets i+1$ \To $j$	\Do
\zi					store $A[m] \gets A[m-1]+1$
				\End	
\zi				$i\gets j+1$
\zi				$\id{change} \gets \const{true}$
			\End
\zi			\If $A'[i \twodots n] \neq A'[A[i] \twodots A[i] + (n-i)]$ \kw{then} 
				$\id{change} \gets \const{true}$
\zi			\If $\id{change}$ \Then
\zi				\If $A[i]=0$ \kw{then} \Error $A'$ is not valid at $n$
\zi				$A[i] \gets $ next largest candidate value, including $0$ 
			\End
			\End
\zi		\Until $\kw{not} \id{change}$
	\Until Armageddon
\end{codebox}

\subsubsection*{Remarks}
While \proc{Validate-$\pi$} produced online a word $w$ over minimal alphabet
such that $\pi_w=A$ this is not the case with \proc{Validate-$\pi'$}. At each time-step
\proc{Validate-$\pi'$} can output a word over minimal alphabet such that $\pi'_w=A'$, but it is not possible
to do so online, as the letters assigned to positions on last slope can change during
the run of \proc{Validate-$\pi$}.

Note that since \proc{Validate-$\pi'$} keeps the function $\pi[1\twodots i+1]$
after reading input $A'[1 \twodots i]$, no changes are required
to adapt it to $g$ validation, where $g(i)=\pi'[i-1]+1$ is the function used in \cite{DuvalLecroqLefebvreopiprim}.

\section*{Open problems}
Two interesting questions remain: is there a real time algorithm for validating $A$ as $\pi$ in the pointer machine model? Is there a linear time algorithm for validating $A'$ as $\pi'$?
The latter probably requires eliminating suffix trees from the construction
or some clever encoding of values of $A'$.
We believe it can be done with better understanding of the underlying word combinatorics.



\bibliographystyle{abbrv}
\bibliography{onlinereconstruction}

\begin{thebibliography}{10}

\bibitem{Alstrup00improvedalgorithms}
S.~Alstrup and J.~Holm.
\newblock Improved algorithms for finding level ancestors in dynamic trees.
\newblock In {\em ICALP 2000, LNCS 1853}, pages 73--84. Springer Verlag, 2000.

\bibitem{AG-book}
A.~Apostolico and Z.~Galil, editors.
\newblock {\em Pattern Matching Algorithms}.
\newblock Oxford University Press, 1997.

\bibitem{Crochemore}
J.~Cl{\'e}ment, M.~Crochemore, and G.~Rindone.
\newblock Reverse engineering prefix tables.
\newblock In {\em STACS}, pages 289--300, 2009.

\bibitem{Crochemore-Rytter-book}
M.~Crochemore and W.~Rytter.
\newblock {\em Text Algorithms}.
\newblock Oxford University Press, 1994.

\bibitem{jewels}
M.~Crochemore and W.~Rytter.
\newblock {\em Jewels of Stringology}.
\newblock World Scientific Publishing Company, 2002.

\bibitem{DuvalLecroqLefebvreopiprim}
J.-P. Duval, T.~Lecroq, and A.~Lefebvre.
\newblock Border array on bounded alphabet.
\newblock {\em Journal of Automata, Languages and Combinatorics}, 10(1):51--60,
  2005.

\bibitem{DuvalLecroqLefebvredlugie}
J.-P. Duval, T.~Lecroq, and A.~Lefebvre.
\newblock Efficient validation and construction of border arrays.
\newblock In {\em 11th Mons Days of Theoretical Computer Science}, pages
  179--189, Rennes, France, 2006.

\bibitem{DuvalLecroqLefebvreladnykod}
J.-P. Duval, T.~Lecroq, and A.~Lefebvre.
\newblock Efficient validation and construction of
  {K}nuth{-–}{M}orris{-–}{P}ratt arrays.
\newblock {\em Conference in honor of Donald E. Knuth}, 2007.

\bibitem{firstlinearverify}
F.~Franek, S.~Gao, W.~Lu, P.~J. Ryan, W.~F. Smyth, Y.~Sun, and L.~Yang.
\newblock Verifying a border array in linear time.
\newblock {\em J. Comb. Math. Comb. Comput.}, 42:223--236, 2002.

\bibitem{KMP}
D.~E. Knuth, J.~H. Morris, Jr., and V.~R. Pratt.
\newblock Fast pattern matching in strings.
\newblock {\em SIAM J. Comput.}, 6(2):323--350, 1977.

\bibitem{McCreight}
E.~M. McCreight.
\newblock A space-economical suffix tree construction algorithm.
\newblock {\em J. ACM}, 23(2):262--272, 1976.

\bibitem{countingdistinct}
D.~Moore, W.~F. Smyth, and D.~Miller.
\newblock Counting distinct strings.
\newblock {\em Algorithmica}, 23(1):1--13, 1999.

\bibitem{MP}
J.~H. Morris, Jr. and V.~R. Pratt.
\newblock A linear pattern-matching algorithm.
\newblock Technical Report~40, University of California, Berkeley, 1970.

\bibitem{Simon}
I.~Simon.
\newblock String matching algorithms and automata.
\newblock In {\em Results and Trends in Theoretical Computer Science}, volume
  812 of {\em LNCS}, pages 386--395. Springer, 1994.

\bibitem{Ukkonen}
E.~Ukkonen.
\newblock On-line construction of suffix trees.
\newblock {\em Algorithmica}, 14(3):249--260, 1995.

\end{thebibliography}

\end{document}